\documentclass[copyright,creativecommons]{eptcs}
\providecommand{\event}{CCA 2010}
\usepackage{amssymb,amsmath,amsthm,graphicx}

\newtheorem{proposition}{Proposition}[section]
\newtheorem{corollary}{Corollary}[section]
\newtheorem{theorem}{Theorem}[section]
\theoremstyle{definition}
\newtheorem{definition}{Definition}[section]
\newtheorem{assumption}{Assumption}[section]

\newcommand{\In}{\subseteq}

\newcommand{\IN}{\mathbb{N}}
\newcommand{\IZ}{\mathbb{Z}}
\newcommand{\IR}{\mathbb{R}}
\newcommand{\ld}{\mathrm{ld}}

\title{Computational Complexity of Iterated Maps on the Interval\\
(Extended Abstract)}
\author{Christoph Spandl
\institute{Institut f\"{u}r Theoretische Informatik, Mathematik und Operations Research}
\institute{Universit\"{a}t der Bundeswehr M\"{u}nchen\\
D-85577 Neubiberg, Germany}
\email{christoph.spandl@unibw.de}
}
\def\titlerunning{Computational Complexity of Iterated Maps on the Interval}
\def\authorrunning{Christoph Spandl}
\begin{document}
\maketitle

\begin{abstract}
The exact computation of orbits of discrete dynamical systems on the
interval is considered. Therefore, a multiple-precision floating point
approach based on error analysis is chosen and a general algorithm is
presented. The correctness of the algorithm is shown and the
computational complexity is analyzed. As a main result, the
computational complexity measure considered here is related
to the Ljapunow exponent of the dynamical system under consideration.
\end{abstract}

%%%%%%%%%%%%%%%%%%%%%%%%%%%%%%%%%%%%%%%%%%%%%%%%%%%%%%%%%%%%%%%%%%%%%%%%%%%%%%%%%%%%%%
\section{Introduction}

Consider a discrete dynamical system $(D,f)$ on some compact interval $D\In\IR$,
called the phase space, given by a function $f:D\to D$, a recursion relation
$x_{n+1}=f(x_n)$ and an initial value $x_0\in D$. The sequence $(x_n)_n$ of
iterates is called the orbit of the dynamical system in phase space corresponding
to the initial value $x_0$. If such a dynamical system is implemented, that is
a computer program is written for calculating a finite initial
segment of the orbit for given $x_0$, care has to be taken in choosing
the appropriate data structure for representing real numbers. Traditionally,
IEEE 754 {\tt double} floating point numbers \cite{ieee87} are used. However,
if the dynamical system shows chaotic behavior, a problem arises. The finite
and constant length of the mantissa of a {\tt double} variable causes round
off errors which are magnified after each iteration step. Only after a few
iterations, the error is so big that the computed values are actually
useless \cite{Mue01}. To put things right, a rigorous method for computations
with real numbers has to be used. In \cite{Bla05}, this issue is addressed
for the logistic map which is also consiedred as a starting point in the next
section. There, the exact real arithmetic in the form of centered intervals
with bounded error terms is used as described in \cite{bl06}. However, the
notation used in \cite{Bla05, bl06} is an algebraic one based on arbitrary
large integers. On the other hand, the aim of the present paper is to keep
the notation as close as possible to the standard in scientific computing
but being precise in the sense of exact real arithmetic. This has as a
consequence first that the basic data type is not an integer as considered
in \cite{Bla05, bl06}, but a floating point number with a definite
mantissa length. Second, the type of error considered here is the
relative error as is standard in floating point arithmetic - in contrast
to the absolute error considered in \cite{Bla05, bl06}. In practice, a
multiple precision floating point library providing floating point
numbers with arbitrary high mantissa length have to be used.
In the following, it is analyzed how the needed
mantissa length behaves in multiple-precision computations of
iterates of discrete dynamical systems. The mantissa length needed
for floating point numbers such that any computed point of the orbit
has a specified and guaranteed accuracy is examined. Therefore,
a precise mathematical framework for floating point computations
has to be established. The main result shows that the
ratio of mantissa length to iteration length in the limit of
iteration length to infinity is related to the Ljapunow exponent.
Comparing, in \cite{Bla05} only the logistic map is
considered explicitly and the connection to the Ljapunow exponent
is not stated, but observed numerically. In the present paper,
this connection is shown mathematically for a general discrete
dynamical system $(D,f)$.
This result also gives some advice for economically designing exact
algorithms simulating one-dimensional discrete dynamical systems.

%%%%%%%%%%%%%%%%%%%%%%%%%%%%%%%%%%%%%%%%%%%%%%%%%%%%%%%%%%%%%%%%%%%%%%%%%%%%%%%%%%%%%%
\section{Roundoff Error, Error Propagation and Dynamic Behavior}

In this section, the discrete dynamical system $(D,f_\mu)$ with $D=[0,1]$ and
$f_\mu:D\to D$, $f_\mu(x):=\mu x(1-x)$ for some control parameter $\mu\in(0,4]$
is investigated. In the literature, the recursion relation $x_{n+1}=f_\mu(x_n)$ is
called the {\em logistic equation} \cite{ce80}. When implementing the logistic
equation on a real computer and demanding to obtain exact values for the
orbit $(x_n)_n$, the analysis of roundoff errors and of error propagation
requires some care. This is due to the fact that for some values of $\mu$
the dynamics is highly chaotic and therefore inaccuracies are magnified
exponentially in time \cite{ce06, hsd04}.

In the following, for a given initial value $x_0$, the true orbit is denoted by
$(x_n)_n$, whereas the really computed orbit, suffering from roundoff errors
and error propagation, is denoted by $(\hat{x}_n)_n$. Note that even
$\hat{x}_0$ may differ form $x_0$ since the conversion to a floating point
number may cause the very first roundoff error. One goal of this section
is to give a rigorous estimation of the total error in dependence of the
iteration step $n$.

Calculating the orbit $(\hat{x}_n)_n$, two types of error are present.
First, error propagation due to the iteration scheme and second the
roundoff error caused by the calculation of $f_\mu$. Now, let $\hat{x}_n$
for some $n\in\IN$ be given. Then the true error after one iteration step
is $\hat{x}_{n+1} - x_{n+1}$. Since in reality not $f_\mu(\hat{x}_n)$ is
calculated but some erroneous approximation $\hat{f_\mu}(\hat{x}_n)$,
the true error can be written as
$\hat{x}_{n+1} - x_{n+1} = \hat{f}_\mu(\hat{x}_n)-f_\mu(x_n)$.
Hence, the true error can be written as a sum
\begin{equation}
\label{main_err}
\hat{x}_{n+1} - x_{n+1} = (f_\mu(\hat{x}_n)-f_\mu(x_n)) +
 (\hat{f_\mu}(\hat{x}_n) - f_\mu(\hat{x}_n))
\end{equation}
of two terms. The first term describes solely the error propagation
while the second term gives exactly the newly produced error
due to the approximate calculation of $f_\mu$.

To handle the exact values of both errors computationally, interval
arithmetic can be used \cite{ah83}. Interval arithmetic can be seen
in the setting here as a special case of the computational model
of TTE \cite{wh00}, which gives a precise notion for describing
computations over the real numbers. Another strongly related model,
which in some sense reflects the situation here more adequate is the
Feasible Real RAM model \cite{BH98}. For the sake of simplicity however,
an interval setting is used here. For any time step $n$, let
the phase point $x_n$ together with its error be represented by two
floating point numbers $x^l_n$ and $x^u_n$ ($x^u_n\geq x^l_n$) with
given mantissa length $m_n$ forming an interval $[x^l_n,x^u_n]$. The
interval is an enclosure of the real value $x_n$, that is
$x_n\in[x^l_n,x^u_n]$. It is straightforward to transform
the interval to a floating point value $\hat{x}_n$ of
mantissa length $m_n$ by setting
\begin{equation}
\label{val_def}
\hat{x}_n:=gl\left(\frac{x^l_n + x^u_n}{2}\right)
\end{equation}
where $gl(.)$ performs the rounding to nearest floating point number.
The absolute error $e_n:=|\hat{x}_n - x_n|$ of $\hat{x}_n$ can
be estimated via the interval length $d_n:=x^u_n-x^l_n$ by
\begin{equation}
\label{err_def}
e_n\leq\frac{1}{2}d_n + r_n
\end{equation}
where $r_n$ is an error introduced by the rounding operation $gl(.)$ in
Equation \ref{val_def}. An upper bound on $r_n$ will be discussed later, for
now it suffices to say that in general it is small compared to $d_n$.

For doing an error analysis of the logistic equation analytically, some
idealizing assumptions are made. First, the value of $\mu$ is assumed to be
given with such a high precision that no interval representation is needed. Second,
only the error propagation is considered caused by the initial error due to rounding
$x_0$ to some floating point number of mantissa length $m$. Third, the value
of $r_n$ in Equation \ref{err_def} is neglected. The recursion relation then reads
in natural interval extension
\begin{align*}
x^l_{n+1} &= \mu x^l_n(1-x^u_n)\\
x^u_{n+1} &= \mu x^u_n(1-x^l_n)
\end{align*}
with the interval length $d_n$  given by the recursion relation
\begin{align*}
d_{n+1} &= x^u_{n+1} - x^l_{n+1} = \mu(x^u_n - x^u_n x^l_n - x^l_n + x^u_n x^l_n)\\
 &= \mu d_n
\end{align*}
with the obvious solution $d_n=\mu^ n d_0$. Finally the absolute error $e_n$
of $\hat{x}_n$ according to Equation \ref{val_def} can be bounded from
above by
\begin{equation}
\label{err_bound}
e_n\leq\frac{1}{2}d_n=\frac{1}{2}\mu^n d_0.
\end{equation}
Note that the above analysis only holds if the natural interval extension
for $f_\mu$ is derived from the formula $\mu x(1-x)$. If it is derived
from the formula $f_\mu(x)=\mu(x-x^2)$ or
$f_\mu(x)=\frac{\mu}{4}-\mu(x-\frac{1}{2})^2$, the mathematical analysis
is more difficult. However, the problems described in the following also
appear, but in some different form.

The aim now is to calculate, for given $N\in\IN$, $p\in\IZ$
and mantissa length $m$, the orbit up to time $N$
with relative error $10^{-p}$. That is, for
$(\hat{x}_n)_{0\leq n\leq N}$ should hold
\begin{equation}
\label{prec_req}
e_n=|\hat{x}_n - x_n|\leq 10^{-p} x_n \leq 10^{-p}.
\end{equation}
The ideal assumptions require the somewhat unreal setting
that the mantissa length has to be set to some finite, but big enough value $m$
for representing $x_0$ and a virtually infinite value $m_\infty$ for doing the
iteration. Finally, some upper bound on $d_0$ is needed. The value of $d_0$
is given as the roundoff error by representing $x_0$ as a floating point number
of mantissa length $m$. For that, the well known estimate
\begin{equation}
\label{round_zero}
d_0\leq 2^{-m+1}x_0\leq 2^{-m+1}
\end{equation}
exists. Combining (\ref{prec_req}), (\ref{err_bound}) and (\ref{round_zero})
gives as a sufficient condition
\begin{equation}
\label{res_cond}
\mu^n\cdot 2^{-m}\leq 10^{-p}
\end{equation}
for $n=0,\dots,N$.

The minimal $m$, fulfilling the precision requirement (\ref{prec_req}) on the
relative error of $x_n$, which depends on $x_0$, $N$ and $p$, is
denoted by $m_{min}(x_0,N,p)$. So, the sufficient condition (\ref{res_cond})
gives an upper bound on $m_{min}(x_0,N,p)$ by
\begin{equation}
\label{ub_req}
m_{min}(x_0,N,p)\leq \lceil p\cdot\ld(10) + N\cdot\max(0, \ld(\mu))\rceil
\end{equation}
where $\ld(.)$ is the logarithm to base $2$. At that stage, a central
quantity of this work is introduced which is a kind of complexity measure.
The {\em loss of significance rate} $\sigma(x,p)$, which may depend on
the initial value $x=x_0$ and the precision $p$ is defined by
\begin{equation*}
\sigma(x,p):=\limsup_{N\to\infty}\frac{m_{min}(x,N,p)}{N}.
\end{equation*}
This quantity describes the limiting amount of significant mantissa length
being lost at each iteration step. Significant means here the part of
the places being exact. A general treatment of this complexity measure
is given in the next section. Roughly speaking,
$\lceil\sigma(x_0,p) N + p\cdot\ld(10)\rceil$ is the mantissa length
for any floating point number needed in an algorithm doing the iteration
starting with $x_0$ and calculating up to $x_N$, if the output should be precise
up to $p$ decimal places. Formula \ref{ub_req} gives an upper
bound for the loss of significance rate by
$\sigma(x,p)\leq\max(0, \ld(\mu))$.

It is interesting to see whether the upper bound calculated analytically,
which needed strong idealizations, is in the region of the real value.
So, the logistic equation was implemented using a multiple-precision
interval library. For that purpose, the interval library MPFI
\cite{ReRo02} based on the multiple-precision floating point number
library MPFR \cite{Fousse:2007:MMP}, both written in C, was used.
For each control parameter $\mu$ ranging from $0.005$ to $4$ and a
step size of $0.005$, the orbit for initial condition $0.22$ was
calculated up to $N=2000$. For each $\mu$, the minimum
mantissa length $m_{min}$ needed to guarantee $e_n\leq 10^{-6}x_n$
for $n=0,\dots,N$ was searched. Then, $\sigma_{est}:=m_{min}/N$
was calculated. The result shows that $\sigma_{est}$ exceeds the
analytical bound $\max(0,\ld(\mu))$ only slightly. So, the above made
ideal assumptions seem to be valid. In \cite{Mue01}, the logistic
equation was also investigated for $\mu=3.75$ using the exact real
arithmetic package iRRAM based on the Feasible Real RAM model
\cite{BH98}. In the paper, also the precision needed
to guarantee the exactness of the first $6$ decimal places
are reported up to $N=100000$. The values are in full
agreement with the simulation results performed here.

Hence, for $\mu > 1$, the interval length $d_n$
increases exponentially in time $n$. This result should be interpreted
in terms of the dynamical behavior of the logistic equation. So, at
this point is worth having an analytical look at the behavior
of the dynamical system. Despite the fact that these results
are well known \cite{hsd04, de89}, they are reviewed here for the sake of
self containment. First, the equation possesses in the range $D=[0,1]$
exactly one
fixed point $x^o=0$ if $\mu\in(0,1]$ and exactly two fixed points $x^o=0$
and $x^{(\mu)}=1-\frac{1}{\mu}$ if $\mu\in(1,4]$. Since $f'_\mu(0)=\mu$
and $f'_\mu(x^{(\mu)})=2-\mu$, $x^o$ is a stable fixed point (an attractor,
$|f'_\mu(x^o)|<1$) for $\mu\in(0,1)$ and an unstable fixed point
(a repeller, $|f'_\mu(x^o)|>1$) for $\mu\in(1,4]$.
If $\mu=1$, the only fixed point $x^o$ is hyperbolic ($|f'_1(x^o)|=1$)
and a bifurcation occurs at that value of the control parameter $\mu$.
If $\mu\in(1,3)$, $x^o$ becomes unstable
and the newly occurring fixed point $x^{(\mu)}$ is stable. Finally,
$\lim_{n\to\infty}f^n_\mu(x)=x^{(\mu)}$ for $\mu>1$ and
$\lim_{n\to\infty}f^n_\mu(x)=x^o$ if $\mu\leq 1$ holds for all $x\in(0,1)$.
If $\mu\in(0,1)$, this is a direct consequence
of the contraction mapping principle. If $\mu=1$, observe that
$f_1(x)<x$ holds for all $x\in(0,1)$. Hence, any sequence $(f^n_1(x))_n$,
$x\in(0,1)$, is strictly decreasing and bounded from below. So it
converges to the only fixed point $x^o$. For the case $\mu\in(1,3)$,
the interested reader is referred to the literature: \cite{de89},
Proposition 5.3. At $\mu=3$ a second bifurcation occurs
and for $\mu>3$ the system goes into a region of periodic
behavior with period doubling bifurcations. Finally, for some
$\mu<4$, chaotic behavior is reached.

This analysis shows that in the parameter range $\mu\in(0,3)$,
the orbit converges to the stable fixed point for any initial value
$x_0\in(0,1)$. Furthermore, there exists some closed interval $I\In D$,
which depends on $\mu$, containing the stable fixed point such that
$f_\mu(I)\In I$ holds and $f_\mu$ is a contraction on $I$. The interval
computation using a natural interval extension of the recursion 
formula $\mu x(1-x)$, on the other hand, is not very compatible with
this picture. While for $\mu\in(0,1)$, the results are
in agreement with the dynamical analysis, the calculations for
$\mu\in(1,3)$ are not handled very well by interval arithmetic since
the interval approach would suggest an exponential divergence of
initially nearby orbits which is not true in reality. The reason is
that the natural interval approach implicitly, due to the dependency problem,
takes account only of the global behavior of $f_\mu$ in the form of the global
Lipschitz constant $\max\{|f'_\mu(x)| : x\in D\}=\mu$. However, the local
Lipschitz constant $\max\{|f'_\mu(x)| : x\in [x^l_n,x^u_n]\}$ governs the
real error propagation at time step $n$ and also describes the dynamic behavior.
This notion can be made precise and finally leads to a more efficient
algorithm for computing orbits.

Let us return to Equation \ref{main_err}. The true error is the sum of
the error propagation (first term) according to the iteration and the
roundoff error due to the computation of $f_\mu$ (second term). The
first term of Equation \ref{main_err} can be handled using the mean
value theorem,
$|f_\mu(\hat{x}_n)-f_\mu(x_n)|=|f'_\mu(y_n)|\cdot |\hat{x}_n - x_n|$
with $y_n\in[\hat{x}_n-e_n,\hat{x}_n+e_n]$. This gives directly the
bound
\begin{equation*}
|f_\mu(\hat{x}_n)-f_\mu(x_n)|\leq
\sup(|f'_\mu([\hat{x}_n-e_n,\hat{x}_n+e_n])|)e_n.
\end{equation*}
The second term can be estimated the following way. As discussed in
\cite{wi63}, the roundoff error produced in calculating $f_\mu$ can be
estimated by
\begin{equation*}
|\hat{f}_\mu(\hat{x})-f_\mu(\hat{x})|\leq 1.06K2^{-m}|f_\mu(\hat{x})|
\end{equation*}
where $K$ is the number of rounding operations performed in computing
$\hat{f}_\mu$ and $m$ is the mantissa length of $\hat{x}$. In
the case considered here, $K=4$ follows since there are 3 arithmetic
operations and the rounding of $\mu$. It is further crucial to
mention that the factor $1.06$ is only valid if $K\leq 0.1\cdot 2^m$
holds so that the mantissa length must not be chosen too small.
Using the fact that $f_\mu(x)\leq\frac{\mu}{4}$ holds and $f_\mu(x)<x$
if $\mu\leq 1$, the unknown value $|f_\mu(\hat{x})|$ can be eliminated.
This calculation shows that there exists a recursive equation
on an upper bound $\overline{e}_n$ on $e_n$ for all $n$:
\begin{equation}
\label{calc_err}
\overline{e}_{n+1}=L(\hat{x}_n,\overline{e}_n)\overline{e}_n+
1.06K2^{-m}E_\mu(\hat{x}_n),\quad \overline{e}_0=2^{-m}
\end{equation}
with $L(x,e):=\sup(|f'_\mu([x-e,x+e])|)$ and
\begin{equation*}
E_\mu(x):=\begin{cases}
x & \text{if } \mu\leq 1\\
\frac{\mu}{4} & \text{if } \mu >1
\end{cases}.
\end{equation*} 
The idea is now not to calculate intervals, but pairs of values $\hat{x}_n$
and corresponding guaranteed error bounds $\overline{e}_n$. The difference to
the interval concept is not to compute the errors {\em implicitly},
so that only global behavior can be taken into account, but to compute
them {\em explicitly} and independent of the values of interest. It should be
mentioned that the approach described here is compatible with an interval
approach using special centered forms, namely mean value forms \cite{rr84}.
However, the approach here explicitly devises values and errors, describes
an automated error analysis, whereas an interval approach primarily does not
disclose any error. Furthermore, also the iRRAM package permits a more
elaborate way for computing the iteration, based on a similar algorithm as
described above \cite{mu10}. The rounded
values $\hat{x}_n$ are calculated as usual in floating point arithmetic except
that multiple-precision floats are used. The guaranteed error bounds are also
calculated using floating point according to (\ref{calc_err}), where
interval arithmetic is used for calculating $L$. Only standard
precision is needed for calculating the error bounds.
Implementing this improved
algorithm using MPFR and MPFI, the setting as given in the interval case
produces the following result. 
In the parameter range $\mu\in(0,3)$, the dynamic behavior is reflected very
well. Furthermore, in the range $\mu\in[3,4]$, the curve suggests
a relation between the loss of significance rate and the Ljapunow
exponent $\lambda(x)$ for the logistic map (for a curve of the Ljapunow
exponent of the logistic map see \cite{ce80}):
$\sigma(x)=\max(0,\lambda(x))/\ln(2)$ for all $\mu\in(0,4]$. To be
complete, the definition of the Ljapunow exponent reads
\begin{definition}
\label{def:ljap}
Let $(D,f)$ be a dynamical system, $D\In\IR$ compact and $f:D\to D$
continuously differentiable on the interior of $D$. Then the
{\em Ljapunow exponent} at $x$ is defined by
\begin{equation}
\label{ljap_def}
\lambda(x):=\lim_{n\to\infty}\frac{1}{n}\sum^{n-1}_{k=0}
\ln\left|f'(f^k(x))\right|
\end{equation}
if the limit exists.
\end{definition}
The Ljapunow exponent may depend on $x$. However, the following
properties hold:
\begin{itemize}
\item[(a)]
If $(D,f)$ has an {\em invariant measure} $\rho$, then the limit in
Equation \ref{ljap_def} exists $\rho$-almost everywhere.
\item[(b)]
Furthermore, if $\rho$ is {\em ergodic} then $\lambda(x)$
is $\rho$-almost everywhere constant and equal to
\begin{equation*}
\int_D \ln\left|f'(x)\right|\,\rho(dx).
\end{equation*}
\end{itemize}
These properties are a direct consequence of the Birkhoff ergodic theorem,
see \cite{kh95}, Theorem 4.1.2 and Corollary 4.1.9.

%%%%%%%%%%%%%%%%%%%%%%%%%%%%%%%%%%%%%%%%%%%%%%%%%%%%%%%%%%%%%%%%%%%%%%%%%%%%%%%%%%%%%%
\section{The General Algorithm and its Complexity}
Let $D$ be a compact real interval and $f:D\to D$ a self mapping. In the following,
$f$ is assumed to be continuous on $D$, continuously differentiable on the interior
of $D$ and $f'$ is bounded. Furthermore, $f$ and $f'$ are assumed to be
computationally feasible. The precise definition of ``computationally feasible''
is given below.

In this section, a general algorithm for computing the iteration
\begin{equation}
\label{main_it}
x_{n+1}=f(x_n),\quad x_0\in D
\end{equation}
is presented. To be more precise, for given $N\in\IN$ and $p\in\IZ$,
this algorithm computes a finite part of the orbit,
$(x_n)_{0\leq n\leq N}$, exact in the sense that the relative
error at each point $x_n$ does not exceed $10^{-p}$.
The correctness of the algorithm and its computational feasibility is shown.
Finally, its complexity is examined.

\subsection{Syntax, Semantics and the Algorithm}
The set of all computationally accessible real numbers are the floating
point numbers of arbitrary mantissa length denoted by $\hat{\IR}$.
In the following, by a floating point number any real number
is meant which can be expressed by normalized scientific notation.
Hence, the set $\hat{\IR}\In\IR$ of all floating point 
numbers is countable infinite and therefore a natural basis for standard
computability considerations. Let $\hat{x}\in\hat{\IR}$ be some floating
point number, then $\hat{x}$ has as an essential property, its
{\em mantissa length} denoted by $\hat{x}.m$.
Any real number $x$ is represented in an algorithm concerning real
computation by a pair $[x]\in\hat{\IR}^2$ consisting of a floating
point number $[x].fl$ approximating $x$ and an upper bound on the
relative error, $[x].err\geq 0$, also being a floating point number.
Furthermore, the inequality $|[x].fl-x|\leq [x].err$ holds. The pair
$[x]$ is called a {\em finite precision representation} of $x$.
Although $[x].err$ has the property mantissa length, it is irrelevant
in what follows. So, the mantissa length of $[x].err$ can be assumed
to be some big enough constant value. Analogously, a function $f:D\to D$,
$D\In\IR$, is called {\em computationally feasible} if a pair
$[f]$ exists of a computable (partial) function $[f].fl:\hat{\IR}\to\hat{\IR}$
approximating $f$ on $D$ and a computable (partial) function
$[f].erf:\hat{\IR}^2\to\hat{\IR}$ giving an upper bound on the
absolute error of $[f].fl$ in the sense
$|[f].fl([x].fl)-f(x)|\leq [f].erf([x])$. Here, a partial function
$\hat{f}:\hat{\IR}\to\hat{\IR}$ is called {\em computable}
if $\hat{f}$ is computable as a string function over some finite
alphabet where the floating point numbers are interpreted as
finite strings. Finally, computability over integers, computability
of functions with mixed arguments and computable predicates are
defined in a standard way.

The algorithm with the above described specification reads
\begin{tabbing}
{\tt\ 1}\ \= {\tt Input parameter:} $\hat{x}_0$, $N$, $p$\\
{\tt\ 2}\> {\tt Initialize mantissa length} $m\leftarrow m_0$\\
{\tt\ 3}\> {\tt do} \= \ \\
{\tt\ 4}\> \> {\tt Initialize value and error} $[x]\leftarrow gl(\hat{x}_0,m)$\\
{\tt\ 5}\> \> {\tt for} \= $n=0$ to $N$ {\tt do}\\
{\tt\ 6}\> \> \> {\tt If} \= $prec([x],p)$ {\tt is {\bf true} then}\\
{\tt\ 7}\> \> \> \> {\tt If not printed print } $n$, $[x].fl$, $[x].err$ \\
{\tt\ 8}\> \> \> {\tt else break}\\
{\tt\ 9}\> \> \> $[x]\leftarrow[f]([x])$\\
{\tt 10}\> \> {\tt end for}\\
{\tt 11}\> \> $[x].fl.m\leftarrow [x].fl.m + 1$\\
{\tt 12}\> {\tt while} $prec([x],p)$ {\tt is false}
\end{tabbing}
To initialize $[x]$, a rounding function
$gl:\hat{\IR}\times\IN\to\hat{\IR}^2$ is needed where $gl(\hat{x}_0, m_0).fl$
is a floating point number of mantissa length $m_0$ being
the exactly rounded value of $\hat{x}_0$ to some rounding convention.
Clearly, $gl(\hat{x}_0, m_0).err$ is an upper bound on the absolute
rounding error, e.g.\ $gl(\hat{x}_0, m_0).err=\frac{1}{2}ulp(\hat{x}_0)$ if
the rounding mode is to nearest. The predicate
$prec:\hat{\IR}^2\times\IZ\to\{{\rm\bf true}, {\rm\bf false}\}$
is a test whether the relative error of $[x]$,
$|[x].fl-x|/|x|$ if $x\neq 0$, is bounded by $10^{-p}$. The
semantics reads: If $[x]\in\hat{\IR}^2$ is a finite precision
representation of $x\in\IR$ and $prec([x],p) = {\rm\bf true}$
holds, then $|[x].fl - x|\leq 10^{-p}|x|$ follows.

In the following, some abbreviations are used occasionally.
The floating point numbers and functions are indicated by a
hat: $\hat{x}:=[x].fl$
and $\hat{f}:=[f].fl$. An over-bar indicates an error bound:
$\overline{e}:=[x].err$ and $\overline{erf}:=[f].erf$. Hence,
$[x]$ is equivalent to $(\hat{x},\overline{e})$ and $[f]$
is equivalent to $(\hat{f},\overline{erf})$.

Finally a remark on optimization. The algorithm
is not optimized in the performance. Otherwise, in Line {\tt 10}
something like $m\leftarrow 2m$ should be used. Here, the aim is to
find the minimal $m$ to guarantee some given upper bound on the
relative error of $x_n$.

\subsection{Feasibility and Correctness}
It is clear, that the rounding function $gl$ is computationally
feasible. So lets begin with the predicate $prec$.
\begin{proposition}
The computationally feasible formula
\begin{equation}
\label{alg:prec}
prec((\hat{x},\overline{e}),p):=\begin{cases}
{\rm\bf true} & \text{if }
\overline{e}\leq\frac{10^{-p}}{1+10^{-p}}|\hat{x}|\\
{\rm\bf false} & \text{else}
\end{cases}
\end{equation}
fulfills the above described semantics.
\end{proposition}
\begin{proof}
Let $(\hat{x},\overline{e})$ be a finite precision
representation of $x$. So, if $\overline{e}\leq 10^{-p}|x|$
holds, then also $|\hat{x} - x|\leq 10^{-p}|x|$ holds.
If $(\hat{x}-\overline{e})(\hat{x}+\overline{e})\geq 0$, then
$|\hat{x}|-\overline{e}\leq |x|$ holds. Hence, if
$(\hat{x}-\overline{e})(\hat{x}+\overline{e})\geq 0$ and
$\overline{e}\leq 10^{-p}(|\hat{x}|-\overline{e})$ holds,
then also $|\hat{x} - x|\leq 10^{-p}|x|$. Finally, if
$\overline{e}\leq\frac{10^{-p}}{1+10^{-p}}|\hat{x}|$
holds, then also
$(\hat{x}-\overline{e})(\hat{x}+\overline{e})\geq 0$.

Formula \ref{alg:prec} only uses the accessible floating point
values $\hat{x}$ and $\overline{e}$, basic arithmetics and
finite tests. Hence, this formula is computationally feasible.
\end{proof}
Note that the definition of the predicate this way
also gives ${\rm\bf true}$ in the singular case where $\hat{x}=0$
and $\overline{e}=0$ and hence $x=0$.

An algorithm for computing $\hat{f}$ is by assumption possible.
To derive an algorithm for computing $\overline{erf}$ on the absolute
error, return to Equations \ref{main_err} and \ref{calc_err}.
\begin{proposition}
\label{prop:ubound}
Assume that $\hat{f}(\hat{x})$ computes $f(\hat{x})$ up to
a correctly rounded last bit in mantissa according to rounding convention.
Then there exists a constant $K>0$ such that the absolute error of
$f(x)$ of the computation $[f]([x])$ is bounded from above by
\begin{equation}
L(\hat{x},\overline{e})\cdot\overline{e} +
\frac{K2^{-m}}{1-K2^{-m}}|\hat{f}(\hat{x})|
\end{equation}
if $K2^{-m}<1$. Here,
$L(\hat{x},\overline{e}):=\sup(|f'([\hat{x}-\overline{e},\hat{x}+
\overline{e}])|)$ and $m$ is the mantissa length of $\hat{x}$: $\hat{x}.m$.

Furthermore, this bound is computable.
\end{proposition}
\begin{proof}
Using Equation \ref{main_err} and following the calculations leading to
equation \ref{calc_err},
$|\hat{f}(\hat{x})-f(x)|\leq L(\hat{x},\overline{e})\cdot\overline{e}+
|\hat{f}(\hat{x})-f(\hat{x})|$ follows. According to the assumption
on $\hat{f}$,
$|\hat{f}(\hat{x})-f(\hat{x})|\leq K2^{-m}|f(\hat{x})|$ holds, with
a value $K\in\{1,2\}$ depending on the rounding convention. However,
$f(\hat{x})$ is unknown, only $\hat{f}(\hat{x})$ is accessible.
To overcome this, set $\hat{f}(\hat{x})-f(\hat{x})=\delta f(\hat{x})$
with $|\delta|\leq K 2^{-m}$. Since $|\delta|<1$ holds, resolve to
$f(\hat{x})=\frac{1}{1+\delta}\hat{f}(\hat{x})$. Hence,
\begin{equation*}
|\hat{f}(\hat{x})-f(\hat{x})|=\left|\frac{\delta}{1+\delta}\right|\cdot
|\hat{f}(\hat{x})|\leq\frac{K2^{-m}}{1-K2^{-m}}|\hat{f}(\hat{x})|
\end{equation*}
follows. Since an upper bound on $L(\hat{x},\overline{e})$ can be
computed using global optimization techniques, e.g.\ with interval
arithmetic, the above described bound is computable.
\end{proof}

To summarize, the mathematical iteration (\ref{main_it}) is performed
in the algorithm by iterating a value $\hat{x}_n$ approximating $x_n$
with an upper bound on its absolute error $\overline{e}_n$ according to
\begin{align}
\label{comp_rec1}
\hat{x}_{n+1} & =\hat{f}(\hat{x}_n) & \hat{x}_0 & =gl(x_0,m)\\
\label{comp_rec2}
\overline{e}_{n+1} & =\overline{L}(\hat{x}_n,\overline{e}_n)\overline{e}_n +
\frac{K2^{-m}}{1-K2^{-m}}|\hat{x}_{n+1}| &
 \overline{e}_0 & =\frac{K2^{-m}}{1-K2^{-m}}|\hat{x}_0|
\end{align}
where $\overline{L}(\hat{x}_n,\overline{e}_n)$ is computable
upper bound on $L(\hat{x}_n,\overline{e}_n)$ as described in the
preceding proposition.
This is Line {\tt 9} in the inner {\tt for}-loop of the algorithm
which is executed with successively increasing mantissa length $m$,
controlled by the outer {\tt do-while}-loop. Finally, it has to
be shown that this outer loop eventually terminates. Therefore,
two more propositions are needed.
\begin{proposition}
Let $x\neq 0$ be a real number and $([x]_m)_{m\geq m_0}$
a sequence of finite precision representations of $x$ with
increasing mantissa length $([x]_m).fl.m \geq m$ such that
$\lim_{m\to\infty}([x]_m).err=0$ holds and consequently
$\lim_{m\to\infty}([x]_m).fl=x$. Then
$\lim_{m\to\infty}prec([x]_m,p)={\bf true}$ follows for
all $p\in\IZ$.
\end{proposition}
\begin{proof}
Since $x\neq 0$ there
exists some $M\in\IN$ such that $0<\frac{1}{2}|x|\leq|([x]_m).fl|$
and $([x]_m).err\leq\frac{10^{-p}}{2(1+10^{-p})}|x|$ holds for all
$m\geq M$. Then, $prec([x]_m,p)={\bf true}$ holds for all
$m\geq M$.
\end{proof}
The next proposition makes the link to Line {\tt 9}
in the algorithm.
\begin{proposition}
\label{prop:alg_halt}
Let $x_n$ be the $n$-th element of the orbit of Equation
\ref{main_it} and $([x_n]_m)_{m\geq m_0}$ a sequence given
according to  the recursion equations (\ref{comp_rec1}) and
(\ref{comp_rec2}) with increasing mantissa length
$([x_n]_m).fl.m \geq m$. Then
$\lim_{m\to\infty}([x_n]_m).err=0$ holds and consequently
$\lim_{m\to\infty}([x_n]_m).fl=x_n$.
\end{proposition}
\begin{proof}
Let  $L:=\sup(f'(D))$ and $\overline{L}\geq L$ be some computationally
accessible value using some global optimization technique. Then
Equation \ref{comp_rec2} leads to 
$\overline{e}_{n+1}\leq \overline{L}\overline{e}_n +
\frac{K2^{-m}}{1-K2^{-m}}\overline{M}$ where
$\overline{M}\geq\sup\{|x| : x\in D\}$ such that
$|\hat{x}_n|\leq\overline{M}$ holds for all $n$.
Iteration gives
$\overline{e}_n\leq\overline{L}^n\overline{e}_0+
\frac{K2^{-m}}{1-K2^{-m}}\overline{M}
\sum^{n-1}_{k=0}\overline{L}^k \leq
\frac{K2^{-m}}{1-K2^{-m}}\overline{M}
\sum^n_{k=0}\overline{L}^k$.
Hence, for $n$ fixed, $\lim_{m\to\infty}([x_n]_m).err=0$ follows.
\end{proof}
These two propositions finish the correctness proof of the
algorithm. They show that, if $x_n\neq 0$ for $n=0,\dots,N$,
the outer loop eventually terminates for any $p\in\IZ$.

\subsection{Computational Complexity}
After having presented the preliminary work, the main issue of the
paper is addressed - the computational complexity of the presented
algorithm. The complexity measure of interest here is the loss of
significance rate already introduced informally in the last section.
Here is the formal definition.
\begin{definition}
\label{def:losr}
The minimal mantissa length, for which the described algorithm
eventually halts is denoted by $m_{min}(x_0,N,p)$, where $x_0$, $N$
and $p$ are the corresponding input parameters. Then, the
{\em loss of significance rate}
$\sigma:\hat{\IR}\cap D\times\IZ\to\IR$ is defined by
\begin{equation}
\sigma(x,p):=\limsup_{N\to\infty}\frac{m_{min}(x,N,p)}{N}.
\end{equation}
\end{definition}
However, to achieve bounds on the loss of significance rate,
a technical difficulty has to be
circumvented. Therefore, one more assumption on the dynamical system
$(D,f)$, additional to the ones already mentioned in the beginning
of this section, has to be made.
\begin{assumption}
\label{ass:main}
The dynamical system $(D,f)$ is assumed to have the properties
already mentioned in the beginning of this section and
additionally $0\not\in D$.
\end{assumption}
It was already seen in the last subsection that $x_n=0$
makes difficulties such that it cannot be proven that
the algorithm eventually halts. However,
the restriction $0\not\in D$ is no loss of generality. If all
other conditions are fulfilled except that $D$
contains zero, consider the following dynamical system
$(\tilde{D},\tilde{f})$ instead. Choose some $M>\min(D)$ and
set $\tilde{D}:=\{x+M \mid x\in D\}$ as well as
$\tilde{f}(x):=f(x-M)+M$ for all $x\in\tilde{D}$. Then
$(\tilde{D},\tilde{f})$ fulfills all required properties.
Furthermore $\tilde{f}'(x)=f'(x-M)$ holds and therefore there
is no substantial difference in the complexity analysis of the
algorithm between the original system and the modified system.

First, the boundedness of $\sigma(x)$ is shown.
\begin{proposition}
Let $(D,f)$ be as in Assumption \ref{ass:main} and $m_{min}(x_0,N,p)$
as in Definition \ref{def:losr}. Then, for given $p\in\IZ$, there exist
some $C_1, C_2\geq 0$, dependent of $f$, such that
$m_{min}(x_0,N,p)\leq C_1 N + C_2$ holds for all
$N\in\IN$, $x\in\hat{\IR}\cap D$.
\end{proposition}
\begin{proof}
According to the requirements made on $(D,f)$, there are some constants
$L>0$ and $M>0$ such that
$\overline{e}_{n+1}\leq L\overline{e}_n+\frac{K2^{-m}}{1-K2^{-m}} M$
holds for all $n\in\IN$ and all mantissa lengths $m$. Without loss
of generality assume $L\neq 1$, otherwise set $L>1$.
Analogous to the treatment in the proof of Proposition
\ref{prop:alg_halt}, iteration gives 
$\overline{e}_N\leq \frac{K2^{-m}}{1-K2^{-m}}M\sum^N_{n=0}L^n=
\frac{K2^{-m}}{1-K2^{-m}}M\frac{L^{N+1}-1}{L-1}$.
Since there exists some $B>0$ with $B\leq|\hat{x}_n|$ for all $n$,
$\overline{e}_N/|\hat{x}_N|\leq\overline{e}_N/B\leq C2^{-m}L^{N+1}$
follows with $C:=MK/(B(1-K2^{-m_0})(L-1))$ where $m_0$ is the initial
mantissa length, Line {\tt 2} in the algorithm. Then, if
$C2^{-m}L^{N+1}\leq \frac{10^{-p}}{1+10^{-p}}$ holds,
$prec((\hat{x}_n,\overline{e}_n),p)={\rm\bf true}$ for all
$n=0,\dots,N$. This leads to
$m_{min}(x_0,N,p)\leq\\max(0,ld(L))N +
max(m_0, \ld(L)+\ld(C) + p\cdot\ld(10) + \ld(1+10^{-p}))$.
\end{proof}
\begin{corollary}
Let $(D,f)$ be as in Assumption \ref{ass:main} and $\sigma(x,p)$ the loss
of significance rate. Then, for given $p\in\IZ$, there exists some
constant $C\geq 0$ such that $\sigma(x,p)\leq C$ holds for all
$x\in\hat{\IR}\cap D$.
\end{corollary}
The treatment has now come to a stage that the main statements of
this paper can be formulated. A lower and an upper bound for the
loss of significance rate is given. Furthermore, these bounds
are strongly related to the Ljapunow exponent $\lambda(x)$ defined
in the previous section.
\begin{theorem}
\label{thm:main}
Let $(D,f)$ be as in Assumption \ref{ass:main}, $\sigma(x,p)$ the loss
of significance rate and $\lambda(x)$ the Ljapunow exponent of
$(D,f)$. Then $\sigma(x,p)\geq\lambda(x)/ln(2)$ holds for all
$x\in\hat{\IR}\cap D$, $p\in\IZ$ if $\lambda(x)$ exists.
\end{theorem}
\begin{proof}
First there are two constants $B,M > 0$ such that $|\hat{x}_n|\geq B$
and $|\hat{x}_n|\leq M$ holds for all $n$.
According to Equation \ref{comp_rec2} and Proposition
\ref{prop:ubound}, $\overline{e}_{n+1}\geq|f'(x_n)|\overline{e}_n$
holds. Iteration gives
$\overline{e}_N\geq\frac{BK2^{-m}}{1-K2^{-m}}\prod^{N-1}_{n=0}|f'(x_n)|$.
Hence, $\frac{\overline{e}_N}{|\hat{x}_N|}\geq
\frac{BK2^{-m}}{M(1-K2^{-m})}\prod^{N-1}_{n=0}|f'(x_n)|$ follows.
A necessary condition for the algorithm to terminate is therefore
$\frac{BK}{M}2^{-m}\prod^{N-1}_{n=0}|f'(x_n)|\leq
\frac{10^{-p}}{1+10^{-p}}$ which gives
$m_{min}(x_0,N,p)\geq\sum^{N-1}_{n=0}\ld(|f'(x_k)|)+p\cdot\ld(10)+\ld(\frac{BK}{M})+
\ld(1+10^{-p})$. Following the definitions of the loss of significance rate
and the Ljapunow exponent, $\sigma(x_0,p)\geq\lambda(x_0)/\ln(2)$ follows.
\end{proof}

Before a realistic upper bound on the loss of significance rate
can be presented, one more definition is needed.
\begin{definition}
Let $\alpha>0$ then define a function $\eta_\alpha:(0,\infty)\to\IR$ by
\begin{equation*}
\eta_\alpha(x):=\begin{cases}
\ln(x) & \text{if } x\geq\alpha\\
\ln(\alpha) & \text{if } x<\alpha
\end{cases}.
\end{equation*}
Furthermore, for any $\alpha>0$ define
\begin{equation*}
\overline{\lambda}_\alpha(x):=\limsup_{n\to\infty}\frac{1}{n}
\sum^{n-1}_{k=0} \eta_\alpha(|f'(f^k(x))|)
\end{equation*}
\end{definition}
\begin{proposition}
For all $\alpha>0$ there exists some constant $C\geq 0$
such that $\overline{\lambda}_\alpha(x)\leq C$
holds for all $x\in D$. Furthermore, if the Ljapunow
exponent $\lambda(x)$ exists,
$\lambda(x)\leq\overline{\lambda}_\alpha(x)$ holds.
\end{proposition}
\begin{proof}
Let $L$ be the Lipschitz constant of $f$ and $\alpha>0$.
Then for all $n\in\IN$,
$\frac{1}{n}\sum^{n-1}_{k=0}\eta_\alpha(|f'(f^k(x))|)\leq
\ln(\max(\alpha,L))$ holds. Hence it follows
$\limsup_{n\to\infty}\frac{1}{n}\sum^{n-1}_{k=0}
\eta_\alpha(|f'(f^k(x))|)\leq\ln(\max(\alpha,L))$.
The second assertion follows from the
fact that $\ln(x)\leq\eta_\alpha(x)$ holds for all $x>0$,
$\alpha>0$. 
\end{proof}
\begin{proposition}
Let $x\in D$ be given. If $\lambda(x)$ exists, then
also the limit
\begin{equation}
\label{ljap_sup}
\lim_{\alpha\searrow 0}\overline{\lambda}_\alpha(x)=:
\overline{\lambda}(x)
\end{equation}
exists and $\overline{\lambda}(x)\geq\lambda(x)$.
\end{proposition}
\begin{proof}
Since $\ln(x)\leq\eta_\alpha(x)\leq\eta_\beta(x)$ holds
for all $x>0$, $0<\alpha\leq\beta$, also
$\lambda(x)\leq\overline{\lambda}_\alpha(x)\leq
\overline{\lambda}_\beta(x)$ follows. So if $\alpha$
converges in a monotonic decreasing way to $0$,
$\alpha\searrow 0$, the assertion follows.

\end{proof}
\begin{theorem}
\label{thm:ubnd}
Let $(D,f)$ be as in Assumption \ref{ass:main}, $\sigma(x,p)$ the loss
of significance rate and $\overline{\lambda}(x)$ as in
(\ref{ljap_sup}). Let $x\in\hat{\IR}\cap D$ be given, then for any
$\varepsilon>0$ there is some $p_0\in\IZ$ such that for all $p\geq p_0$,
$\sigma(x,p)\leq\overline{\lambda}(x)/ln(2)+\varepsilon$ holds
if $\overline{\lambda}(x)$ exists.
\end{theorem}
The proof is similar to the proof of Theorem \ref{thm:main}
and can be found in the full version of this article \cite{sp10}.

%%%%%%%%%%%%%%%%%%%%%%%%%%%%%%%%%%%%%%%%%%%%%%%%%%%%%%%%%%%%%%%%%%%%%%%%%%%%%%%%%%%%%%
\section{Conclusions}
In this paper, two main issues are addressed. First it is
shown that a mathematically precise treatment
of multiple-precision floating point computability is not hard
to do. Furthermore this
treatment is in a manner which is familiar to people working
in the field of numerical analysis or scientific computing and
also for theoretical computer scientists. Furthermore, the
formalism does not only allow exact answers concerning the
existence of a computationally feasible algorithm, but is also
allows a treatment of its complexity. As a consequence, the
described algorithm is formulated not only in an exact and
guaranteed way, but also enables a motivated reader the real
implementation and gives a practical performance analysis.

Second, the results show that the Ljapunow exponent, a central
quantity in dynamical systems theory, also finds its way into
complexity theory, a branch in theoretical computer science.
In dynamical systems theory, the Ljapunow exponent describes
the rate of divergence of initially infinitesimal nearby
points. For two points having a small but finite initial
separation, the Ljapunow exponent has only relevance for
short time scales \cite{ce06}. The reason is that due to the
boundedness of $D$, any two different orbits cannot separate
arbitrarily far away. However, the loss of significance rate
shows that the Ljapunow exponent has on long time scales
not only an asymptotic significance but also a concrete
practical one.

\subsection*{Acknowledgments}

The author wishes to express his gratitude to Peter Hertling for helpful
discussions and comments.

%\bibliography{../../biblio/bib-all,../../biblio/bibliography}
\bibliographystyle{eptcs}

\end{document}